\documentclass[conference,10pt]{IEEEtran}
\usepackage{epsfig,rotating,setspace,latexsym,amsmath,amssymb,amsfonts,bm,theorem,subcaption,epstopdf}
\usepackage{cite,authblk}
\usepackage{bbm}
\usepackage{algorithmic,algorithm}
\newtheorem{theorem}{Theorem}

\newenvironment{proof}[1]{\medskip\par\noindent{\bf
Proof:\,}\,#1}{{\mbox{\,$\blacksquare$}\par}}
\IEEEoverridecommandlockouts
\newtheorem{corollary}{Corollary}
\allowdisplaybreaks
\usepackage{color}
\usepackage{graphicx}

\begin{document}

        \title{Age of Information in G/G/1/1 Systems}

        \author[1]{Alkan Soysal}
        \author[2]{Sennur Ulukus}
        \affil[1]{\normalsize Department of Electrical and Electronics Engineering,
Bahcesehir University, Istanbul, Turkey}
        \affil[2]{\normalsize Department of Electrical and Computer Engineering,
University of Maryland, MD}

        \maketitle

\begin{abstract}
We consider a single server communication setting where the interarrival times of data updates at the source node and the service times to the destination node are arbitrarily distributed. We consider two service discipline models. If a new update arrives when the service is busy, it is dropped in the first model; and it preempts the current update in the second model. For both models, we derive exact expressions for the age of information metric with no restriction on the distributions of interarrival and service times. In addition, we derive upper bounds that are easier to calculate than the exact expressions. In the case with dropping, we also derive a second upper bound by utilizing stochastic ordering if the interarrival times have decreasing mean residual life (DMRL) and service times have new better than use in expectation (NBUE) property. 
\end{abstract}

\section{Introduction}
No matter how important information might be, there is a duration of time after which information loses its freshness. Especially in today's world of immensely interactive everything, information ages fast. Hence, in recent years, researchers started to consider the age of information (AoI) as well as the value of information. Age of anything can be defined as the duration between the time of birth and now. This definition is sufficiently broad to cover all communication scenarios. Nevertheless, most of AoI literature considers queuing systems with well-behaved distributions. In this paper, we take this a step forward and apply AoI viewpoint to more general communication scenarios. 

In a communication setting, the first papers that consider AoI are\cite{YatesInfocom2012}, \cite{YatesCISS2012}, and \cite{YatesISIT2012}. The authors assume First Come First Served (FCFS) systems and calculate the average AoI expressions for M/M/1, M/D/1 and D/M/1 queues in \cite{YatesInfocom2012}, assume Last Come First Served (LCFS) systems with and without preemption and calculate the average AoI expression for M/M/1 queue in \cite{YatesCISS2012}, assume multi-source FCFS systems with M/M/1 queue in \cite{YatesISIT2012}, and provide a more detailed analysis in \cite{YatesKaul-ArXiv}. Starting with these works, there has been a growing interest in AoI analysis. For example, the authors in \cite{CostaEphr-TransIT} consider a packet management approach for M/M/1/1 and M/M/1/2 queues. \cite{YatesISIT2017} calculates age for an M/G/1/1 queue and finds the optimum arrival rate for minimum age. 

While the literature on calculating age expressions for different queuing models expands, another line of research applies AoI approach to energy harvesting problems. The goal is to find the optimum update generation policy that minimizes age, given the service time distribution. In \cite{SunTransIT2017}, the authors show the existence of optimal stationary deterministic update generation policy when the service time process is a stationary and ergodic Markov chain. Application of AoI to offline energy harvesting is considered in \cite{UysalITA2015}, \cite{ArafaUlukusGC17}, \cite{ArafaUlukusAsi17}, and online energy harvesting is considered in \cite{WuYangWu18}, \cite{ArafaUlukusITA18}, \cite{ArafaUlukusICC18}, \cite{BakninaUlukusISIT18}. 

In this paper, our goal is to analyze AoI for general communication scenarios with arbitrary service time and  interarrival time distributions. Using queuing theory terminology, our model corresponds to a G/G/1/1 system. An example of such a G/G/1/1 system is the multicast problem in \cite{Yates-multicast}, where a new update is generated when a percentage of the destinations receives the current update, and service time to each destination is a shifted exponential random variable. Although an exact expression for their model is derived in \cite{Yates-multicast}, in general, calculating an exact age expression for non-exponential interarrival times is difficult to obtain. For example, \cite{Buyukates18} considers a two-stage multicast extension of \cite{Yates-multicast}, where an upper bound is derived for the age of the second stage nodes. In this paper, we derive exact age expressions given that the distribution of service time is arbitrary but known. In addition, we also derive upper bounds to age of information that might be more practical to utilize. If one designs age-minimizing policies using our upper bounds, resulting age will be an achievable age. 

We consider two service disciplines. The first one is called G/G/1/1 {\em with dropping}, where a new arrival is dropped if the server is busy. This model is also used in \cite{CostaEphr-TransIT} for an M/M/1/1 system and in \cite{YatesISIT2017} for an M/G/1/1 system. Here, we restrict ourselves to neither exponential interarrival times nor exponential service times. We derive an exact expression and two upper bounds for our first service model. While the first upper bound does not have any restrictions, the second upper bound requires the interarrival times to have decreasing mean residual life (DMRL) and service times to have new better than used in expectation (NBUE) property \cite{book/StochasticOrders}. Many distributions, including exponential, Rayleigh, Erlang, and gamma, that appear in arrival processes, have both DMRL and NBUE properties \cite{book/MO-life}. 

Our second service discipline model is called G/G/1/1 {\em with preemption in service}, where a new arrival preempts any ongoing service. This model is used in \cite{YatesInfocom2012} and \cite{YatesKaul-ArXiv}, but for exponential interarrival and service times. The exact expression and upper bound that we derive for this model does not have any restrictions. As numerical examples, we simulate our system models with arbitrary distributions, and compare the simulated age values to calculated exact age and upper bound expressions. We observe that for most of the parameter range, upper bounds are close to the exact age values. 

\section{System Model}

We consider a communication scenario where the data arrive at the source according to an arrival process with independent and identically distributed (i.i.d.) interarrival times $Y_n$. The source transmits the data through a single communication link (server). Time duration of service is modeled as a random process with i.i.d. service times $S_n$. We specify general probability distributions on the interarrivals and service times. 

In Figs.~\ref{fig/geometry-dropping} and \ref{fig/geometry-preemption}, circles correspond to packet arrivals at the source. The interval where the system is idle (no packets in the service or in the queue) is denoted by $W_n$, the service time is denoted by $S_n$, and the interval from the end of one idle interval to the end of the next idle interval is denoted by $G_n$. We will consider two service/queue disciplines. 

\subsection{Dropping}
\label{sec/Model-dropping}
\begin{figure}[t]
     \centering
          \includegraphics[width=.4\textwidth]{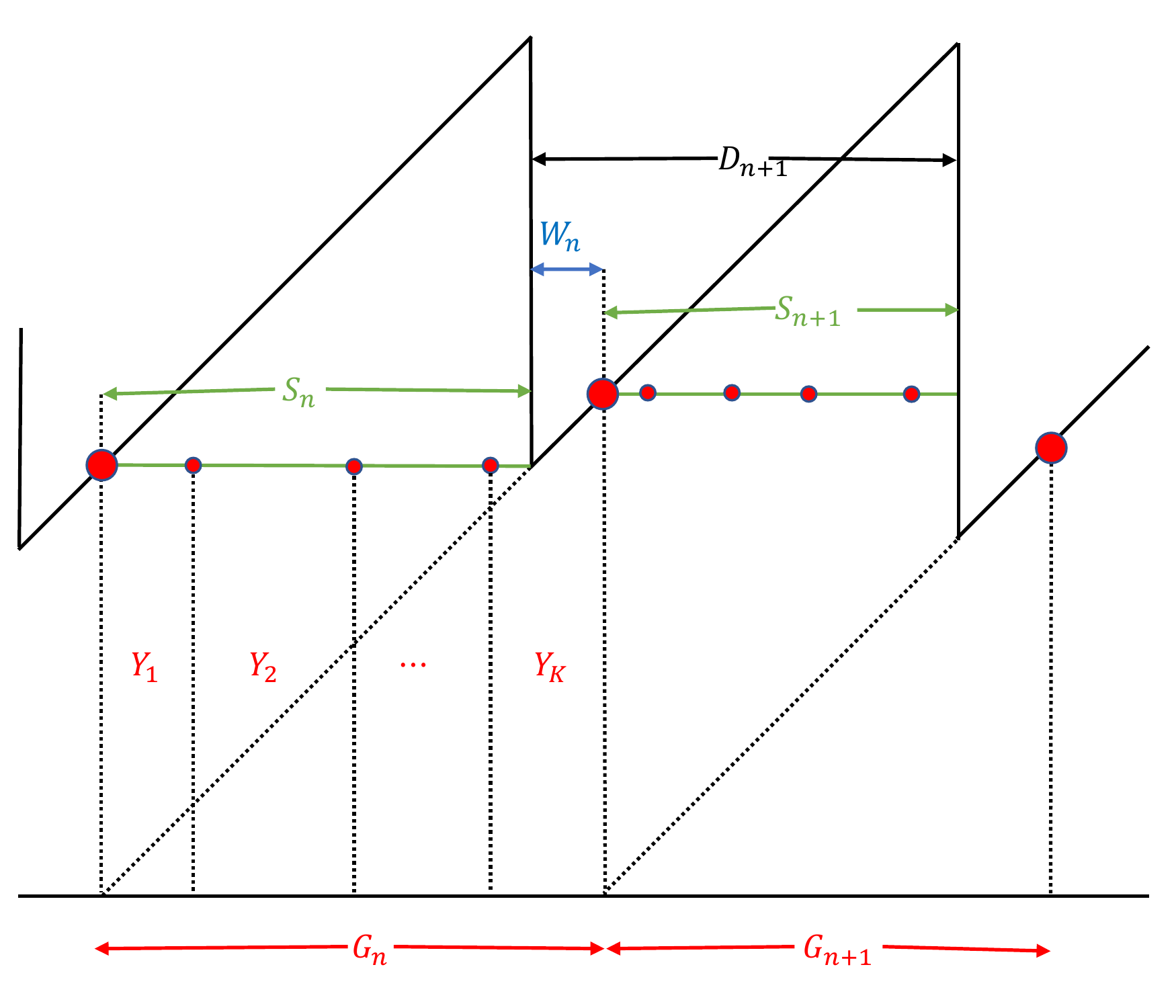}
     \caption{Age curves for G/G/1/1 with dropping model.}
     \label{fig/geometry-dropping}
\end{figure}

In this model, if an update arrives while the server is busy, it is dropped. If an update arrives while the server is idle, service is initiated immediately. We refer to those arrivals that initiate a service as the successful arrivals. In Fig.~\ref{fig/geometry-dropping}, successful arrivals for this model are shown with larger circles, while unsuccessful arrivals are shown with smaller circles. Interarrival times between successful arrivals, ${G}_n$, are called effective interarrival times. In this model, system busy time is equal to the service time, $S_n$. After a service is completed, there is a random nonnegative waiting time for a new successful update, which is $W_n = {G}_n -S_n$.

Note that, effective interarrival time, $G_n$, can be written as a random sum of random numbers, $G_n = \sum_{k=1}^K Y_k$, where $K$ is a discrete random variable with 
\begin{align}
\text{Pr}(K=k) = \text{Pr}\left( \sum_{j=1}^{k-1} Y_j \leq S < \sum_{j=1}^k Y_j \right).
\end{align}

\subsection{Preemption in service}
\begin{figure}[t]
     \centering
          \includegraphics[width=.4\textwidth]{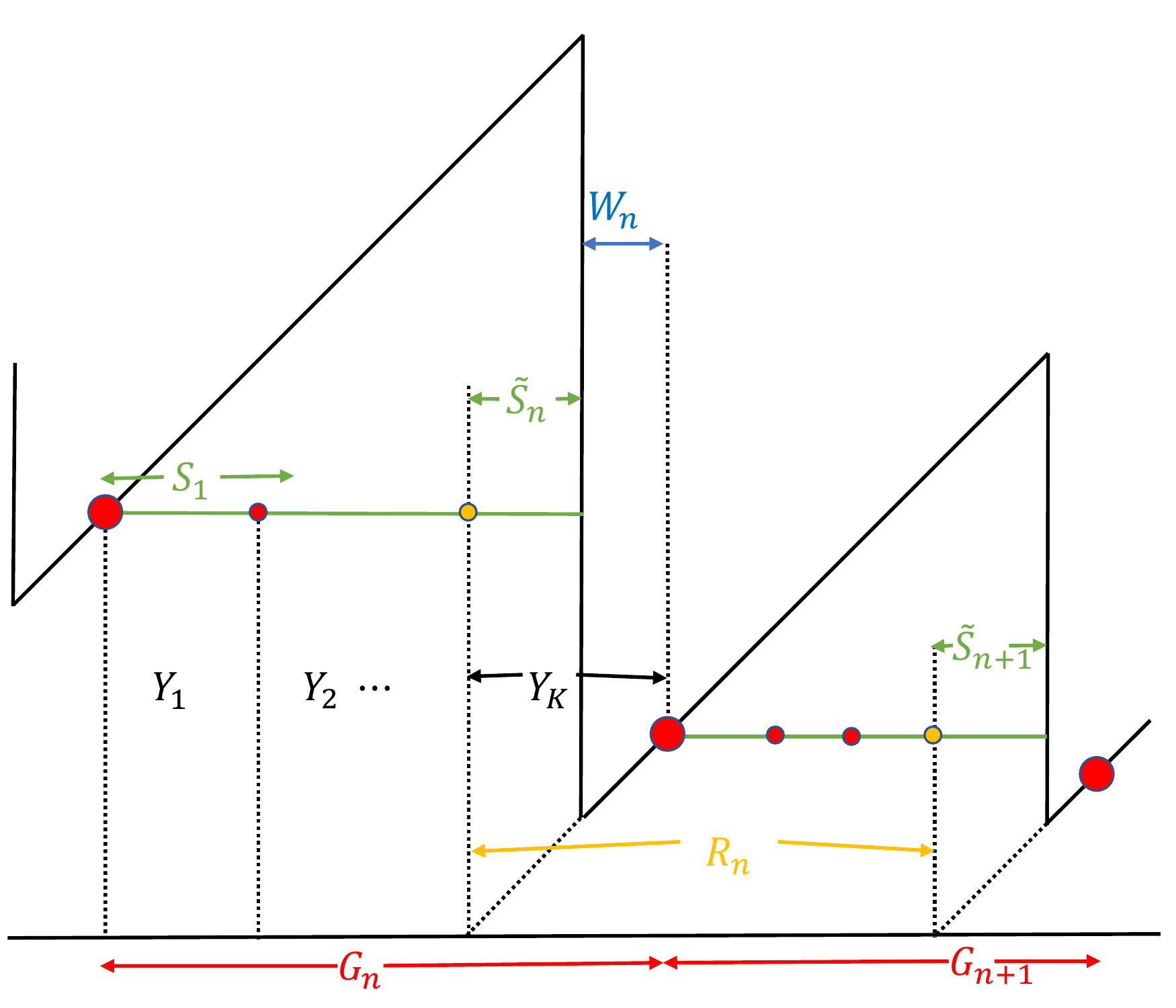}
     \caption{Age curves for G/G/1/1 with preemption in service.}
     \label{fig/geometry-preemption}
\end{figure}

In this model, if an update arrives while the server is idle, the service is initiated immediately. If an update arrives while the server is busy, the packet being served is terminated and the new update is pushed to the server. A successful arrival is the one that can finish the service. In Fig.~\ref{fig/geometry-preemption}, successful arrivals for this model are shown with yellow circles, while unsuccessful arrivals are shown with red circles. Similar to the dropping model, we call the interarrival times between successful arrivals as the effective interarrival times. Note that effective interarrival times, $R_n$, are statistically equal to ${G}_n$. On the other hand, in this model, system busy time is not equal to service time. Service time, $\tilde{S}_n$, is the interval between a yellow circle and the end of busy time. It is important to note that busy time might include interarrival times for unsuccessful arrivals, and $\tilde{S}_n$ is the service time conditioned that service time is smaller than the last interarrival time,  $\tilde{S}_n = S | S < Y$. 

After a service is completed, there is a random nonnegative waiting time for a new successful update. Unlike the previous section, waiting time depends only on the previous arrival time, $W_n = Y_K - S_K | Y_K > S_K$ or equivalently $W_n = Y_K - \tilde{S}_n$, where $K$ is the integer random variable that describes the total number of arrivals that is needed before the next successful arrival. Although $K$ in dropping discipline did not follow a specific distribution, $K$ in preemption in service discipline is a geometric random variable. Before a successful arrival occurs with probability $p=\text{Pr}(Y>S)$, there must be $K-1$ unsuccessful arrivals, all with the same probability $1-p$.

\section{G/G/1/1 with dropping}

For G/G/1/1 with dropping discipline, average age can be written as the difference of the areas of two triangles, divided by the expected value of the effective interarrival time. From Fig.~\ref{fig/geometry-dropping}, we have
\begin{align}
\Delta_\text{GG} =& \frac{E[({G}_n+S_{n+1})^2] - E[(S_{n+1})^2]}{2E[{G}_n]}   \\
=& \frac{E[{G}^2]}{2E[{G}]} + E[S],\label{eq:age-G-dropping}
\end{align}
where $S_{n+1}$ is independent of $G_n$, and time indices are dropped. For most general arrival and service time models, it is not easy to calculate the first and second moments of effective interarrival times. In the following, we first derive an exact expression for (\ref{eq:age-G-dropping}). Then, we also derive upper bounds to (\ref{eq:age-G-dropping}) that is easier to calculate. We will introduce two upper bounds, the first of which is valid for arbitrary interarrival and service times, and the second of which is valid for DMRL interarrivals and NBUE service times.

\subsection{Arbitrary interarrival and service times}

In this section, we start with deriving an exact age expression for the dropping model.
\begin{theorem}
Consider a G/G/1/1 system with dropping discipline, where $Y_n$ are interarrival times and $S_n$ are service times. The average age of this system is given by
\begin{align}
\hspace{-12pt}\Delta_\text{GG} =& \frac{E[Y^2]}{2E[Y]} + \frac{\sum_{k=1}^\infty E[A_k\bar F_S(A_k)]}{E[K]} + E[S]  \hspace{-12pt}
\label{eq:age-exact-dropping}
\end{align}
where $K$ is the average number of arrivals during an effective interarrival time, $A_{k} = \sum_{l=1}^{k-1}Y_l$, and $\bar F_S(\cdot)$ is the complementary cdf of $S$.
\label{th:exact-dropping}
\end{theorem}
\begin{proof}
Remember from Section~\ref{sec/Model-dropping} that effective interarrival times, $G_n$, can be written as random sums of random numbers. Although $K$ is not independent of all $Y_j$, it is possible to calculate the expected value of $G$ using Wald's equation \cite[Theorem 3.3.2]{book/Ross} as $E[G] = E[K]E[Y]$. 

Next, we derive an expression for the second moment of effective interarrival times, $E[G^2]$. Let us first define the indicator function as
\begin{align}
I_k = \left\{
\begin{array}{ll}
1 \quad & \text{if  } k \leq K \\
0 \quad & \text{if  } k > K.
\end{array}
\right.   
\label{eq:indicator}
\end{align}
Now, we have 
\begin{align}
\!\!\!\!E\left[ \!\left( \sum_{k=1}^K Y_k \right)^2\! \right] \! &=\!E\left[ \left( \sum_{k=1}^\infty Y_k I_k\right)^2 \right]    \\
\!&=\! \sum_{k=1}^\infty E[Y_k^2 I_k] \! +\! 2 \sum_{k=1}^\infty\sum_{l=1}^{k-1} E[ Y_k I_k Y_l I_l ].   \label{eq:7}
\end{align}
Note that, $I_k =1$ if and only if, we have not stopped after successively observing $Y_1, \dots, Y_{k-1}$. Therefore, $I_k$ is determined by $Y_1, \dots, Y_{k-1}$, and is thus independent of $Y_k$. We have $E[Y_k^2 I_k] = E[Y_k^2] E[I_k]$, and $E[ Y_k I_k Y_l I_l ] = E[ Y_k]E [I_k Y_l I_l ]$, for $l<k$. Now, (\ref{eq:7}) becomes
\begin{align}
E\left[ G^2\right] \!&=\! E\left[Y^2\right] \sum_{k=1}^\infty E[I_k] \! +\! 2E[Y] \sum_{k=1}^\infty\sum_{l=1}^{k-1} E[ I_k Y_l I_l ].  \label{eq:8}
\end{align}
Let us calculate 
\begin{align}
\hspace{-8pt}\sum_{l=1}^{k-1} E [I_k Y_l I_l]  &= \sum_{l=1}^{k-1}E [Y_l | I_k = 1]\text{Pr}(I_k = 1)   \\
&= E \left[\sum_{l=1}^{k-1}Y_l \, \rule[-12pt]{.45 pt}{28pt} \, I_k = 1\right]\text{Pr}(I_k = 1) \\
&= E \left[\sum_{l=1}^{k-1}Y_l \, \rule[-12pt]{.45 pt}{28pt} \, \sum_{l=1}^{k-1}Y_l \!<\! S\right]\text{Pr}\left(\sum_{l=1}^{k-1}Y_l \!<\! S\right) \label{eq:11}
\end{align} 
where we used the fact that $I_l = 1$ for $l < k$ and given $I_k = 1$. The condition $I_k=1$ and $\sum_{l=1}^{k-1}Y_l < S$ are equivalent for dropping service discipline. Next, let us denote $A_{k} = \sum_{l=1}^{k-1}Y_l$. Then using Bayes' rule, we can calculate that
\begin{align}
E \left[A_k | A_{k} < S\right] = &\int_0^\infty a \frac{\text{Pr}(A_k < S | A_k = a)}{\text{Pr}(A_k<S)}f_{A_k}(a) da \\
&= \frac{E[A_k\bar F_S(A_k)]}{\text{Pr}(A_k<S)}
\end{align}
where $\text{Pr}(A_k < S | A_k = a) = \text{Pr}(S > a) = \bar F_S(a)$. Now, (\ref{eq:8}) becomes
\begin{align}
\hspace{-10pt}E\left[ G^2\right] &= E\left[Y^2\right] E[K]  + 2E[Y] \sum_{k=1}^\infty E[A_k\bar F_S(A_k)]
\label{eq:dummy}
\end{align}
where $\sum_{k=1}^\infty E[I_k] = E[K]$. Now, we have (\ref{eq:age-exact-dropping}), when (\ref{eq:dummy}) is inserted in (\ref{eq:age-G-dropping})
\end{proof}

The exact expression in (\ref{eq:age-exact-dropping}) requires a calculation of an infinite sum. In order to reduce the complexity of calculation, we derive the following upper bound.
\begin{corollary}
Consider a G/G/1/1 system with dropping discipline, where $Y_n$ are interarrival times and $S_n$ are service times. The average age of this system is upper bounded by
\begin{align}
\hspace{-12pt}\Delta_\text{GG} \leq& \frac{E[Y^2]}{2E[Y]} + E[Y]\left(\frac{E[K^2]}{2E[K]}-\frac{1}{2}\right) + E[S]  \hspace{-12pt}
\label{eq:age-ub-dropping}
\end{align}
where $K$ is the average number of arrivals during an effective interarrival time.
\label{co:ub-dropping}
\end{corollary}
The proof of Corollary~\ref{co:ub-dropping} follows by noting that
\begin{align}
E \left[\sum_{l=1}^{k-1}Y_l \; \rule[-12pt]{.45 pt}{28pt} \; \sum_{l=1}^{k-1}Y_l < S\right] &\leq E \left[\sum_{l=1}^{k-1}Y_l\right],
\label{eq:dummy2}
\end{align}
and $\sum_{k=1}^\infty (k-1) E[I_k]  = \frac{1}{2}E[K(K-1)]$. In our system model, the number of terms in the random sum, $G$, depends on the summands, $Y_k$. Under this system model, the upper bound in (\ref{eq:age-ub-dropping}) can only be achieved with deterministic interarrival times, $Y_k$. On the other hand, when (\ref{eq:dummy2}) is applied to (\ref{eq:dummy}), we get 
\begin{align}
\!\!\!E[ G^2] \leq & E[Y^2] E[K]  + (E[Y])^2 (E[K^2] - E[K])\label{eq:dummy3}
\end{align}
where the right hand side is equal to the second moment of a random sum when the number of terms in the sum is independent of the summands \cite{book:Yates}.  

It is important to note that Corollary~\ref{co:ub-dropping} reduces the complexity of calculation. For exponential service times, we have a closed form expression for the upper bound to the average age. It can be shown that for exponentially distributed service times with parameter $\mu$, $K$ is a geometric random variable with $p = 1-E[e^{-\mu Y}]$. Then, the upper bound to the age of G/M/1/1 systems can be written as
\begin{align}
\hspace{-12pt}\Delta_\text{GM} \leq & \frac{E[Y^2]}{2E[Y]} + E[Y]\left(\frac{1}{1-E[e^{-\mu Y}]}- 1 \right) + \frac{1}{\mu} . \hspace{-12pt}
\label{eq:age-ub-GM-dropping}
\end{align}
Since all the components in (\ref{eq:age-ub-GM-dropping}) are known, one can use this upper bound to design age-minimal policies for communication systems with exponential service time and without a restriction on the arrival process. We remark that the resulting age of such an optimization is guaranteed to be achieved. 

When the interarrival times are also exponentially distributed with rate parameter $\lambda$, (\ref{eq:age-ub-GM-dropping}) becomes
\begin{align}
\Delta_\text{MM} \leq & \frac{1}{\lambda} + \frac{2}{\mu}
\end{align}
where the age of M/M/1/1 queues is known to be $\Delta_\text{MM} =\frac{1}{\lambda} + \frac{2}{\mu} - \frac{1}{\lambda+\mu}$ \cite{CostaEphr-TransIT}.

\subsection{DMRL interarrival and NBUE service times}

For some communication models, it might not be possible to calculate the moments of $K$ that is needed in (\ref{eq:age-ub-dropping}) for arbitrary interarrival times. For such cases, we derive another upper bound that does not include those moments. This bound requires the interarrival times to have DMRL and service times to have NBUE property \cite{book/StochasticOrders}. 

\begin{theorem}
Consider a G/G/1/1 system with dropping discipline, where $Y_n$ are i.i.d. with arbitrary distribution that has DMRL property, and $S_n$ are i.i.d. with arbitrary distribution that has NBUE property. Let us also consider an M/G/1/1 system that is formed by replacing the interarrival times of G/G/1/1 system with exponentially distributed interarrival times, $Y^e_n$, where $E[Y^e]=E[Y]$. Then, the average age of G/G/1/1 system, $\Delta_\text{GG}$ is upper bounded by the average age of M/G/1/1 system, $\Delta_\text{MG}$. 
\label{th:mg-dropping}
\end{theorem}
\begin{proof}
In order to prove this theorem, we need several results from stochastic ordering \cite{book/StochasticOrders}. Here, we will only provide the outline of the proof due to space limitations. Since $D_n$ is equally distributed with $G_n$, and ${D}_{n+1} = W_n + S_{n+1}$, (\ref{eq:age-G-dropping}) can be written as
\begin{align}
\Delta_\text{GG} =& \frac{E[(W_n+S_{n+1})^2]}{2E[(W_n+S_{n+1})]} + E[S].
\label{eq:age-dropping-w+s}
\end{align}
Our goal is to show that the expression in (\ref{eq:age-dropping-w+s}) is less than equal to
\begin{align}
\Delta_\text{MG} =&  \frac{E\left[ (Y^e+S)^2\right]}{2E[(Y^e+S)]} + E[S]
\label{eq:age-ub-MG-dropping}
\end{align}
where the right hand side is the age for M/G/1/1 queues \cite{YatesISIT2017}. We first argue that interarrival time random variable, $Y$, with DMRL property for the G/G/1/1 system is smaller (in the convex order) than the exponential random variable with the same mean, $Y^e$, for the M/G/1/1 system. Therefore, we have $Y\leq_\text{cx}Y^e$. 

We can write the waiting time as $W=(Y_{K} - X | Y_{K} > X)$, where $X$ is the remaining service after the last unsuccessful arrival. Note that $X$ depends on $S$ and $Y_{j}, j=1,\dots, K-1$, but it is independent of $Y_{K}$. Using \cite[Section 11-4, eqn. (50)]{book/Wolff} for DMRL interarrival times, we have $W\leq_\text{hmrl}Y\leq_\text{hmrl} Y^e$. Moreoever, since $S_{n+1}$ is independent of $W_n$ and $Y^e_n$, from \cite[Lemma 2.B.5]{book/StochasticOrders}, we have $W+S\leq_\text{hmrl} Y^e+S$. Finally, using \cite[eqn. (2.B.5)]{book/StochasticOrders}, we have
\begin{align}
\frac{E[(W+S)^2]}{E[W+S]} \leq \frac{E[(Y^e+S)^2]}{E[Y^e+S]},
\end{align}
which directly implies $\Delta_\text{GG} \leq \Delta_{MG}$.
\end{proof}
\begin{figure}
	\begin{subfigure}[t]{0.5\textwidth}
		\centering
		\includegraphics[width=3.3in]{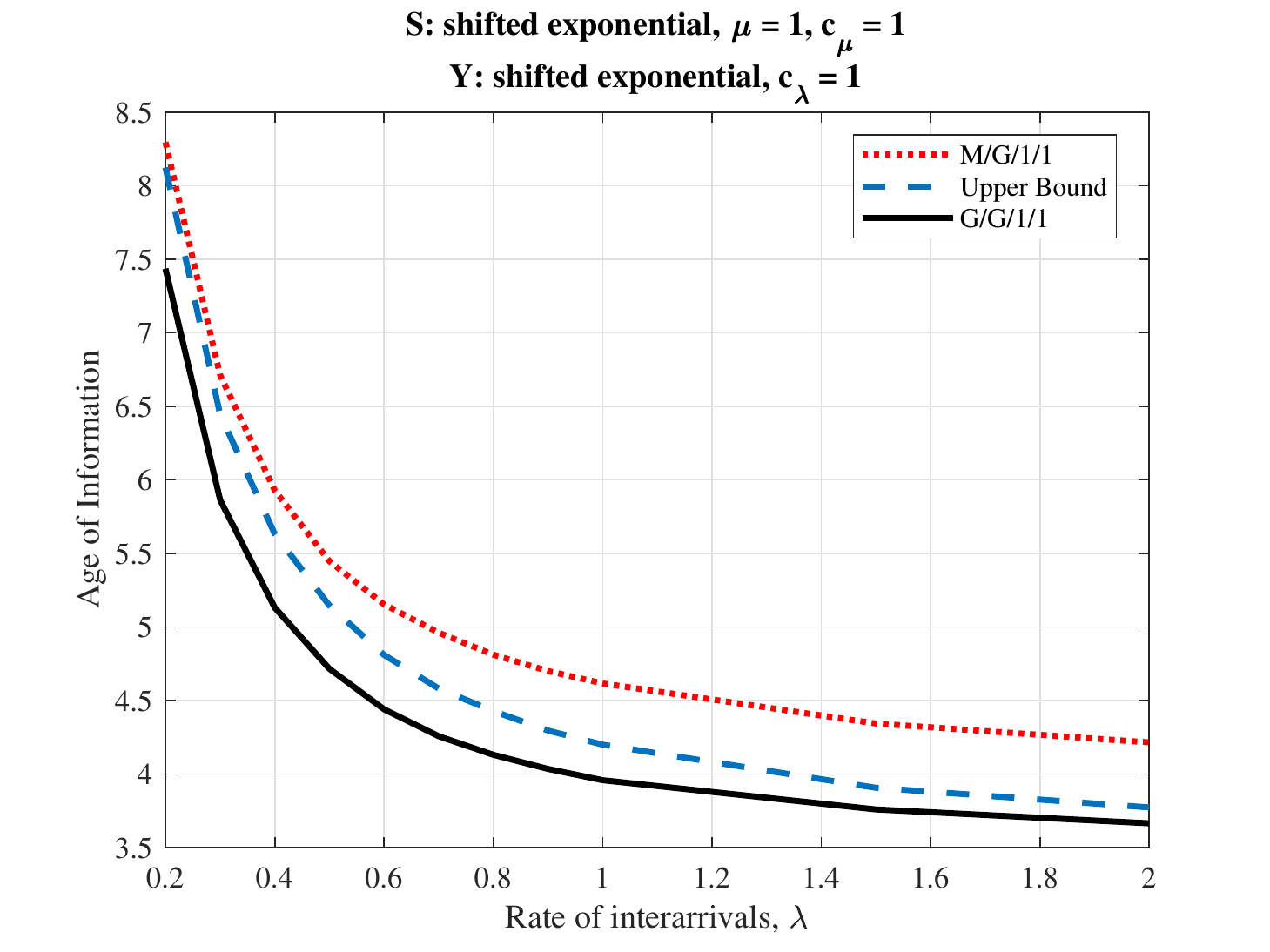}
		\vspace{-5pt}
\caption{}
	\end{subfigure}
    \begin{subfigure}[t]{0.5\textwidth}
		\centering
		\includegraphics[width=3.3in]{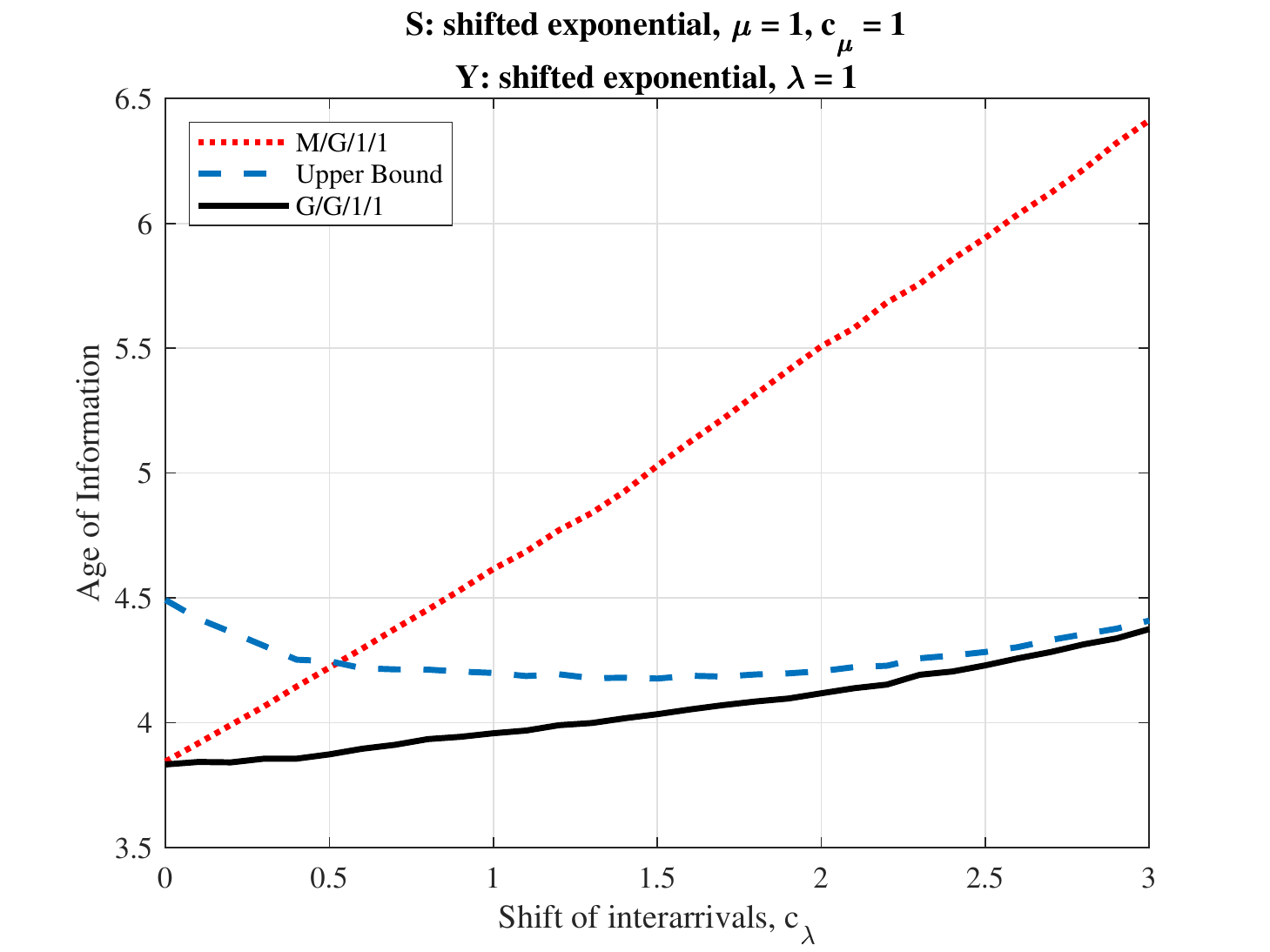}
		\vspace{-5pt}
\caption{}
	\end{subfigure}
	\caption{Average AoI for G/G/1/1 with dropping.}\label{fig/SE}
\end{figure}

In order to observe the tightness of our upper bounds, we simulate an example G/G/1/1 system, calculate its age and compare it to the proposed upper bounds. As in \cite{Yates-multicast}, we consider shifted exponential interarrival times with rate parameter $\lambda$ and shift parameter $c_\lambda$; and shifted exponential service times with rate parameter $\mu$ and shift parameter $c_\mu$. In Fig.~\ref{fig/SE}, we observe that age decreases with the rate parameter and increases with the shift parameter of the interarrival distribution. The distance between the first proposed upper bound and the exact age for G/G/1/1 seems to be bounded and small for both curves. When the shift of interarrival times is zero, G/G/1/1 reduces to M/G/1/1. As the shift of interarrival times increase, upper bound converges to the exact age for G/G/1/1. The reason for this is that as $c_\lambda$ increases, coefficient of variation of the interarrival times decreases, and therefore interarrival times appraoches to a deterministic value.

In Fig.~\ref{fig/imrl}, we consider the effect of having Increasing Mean Residual Life (IMRL) interarrival times. We observe that while the expression in Theorem~\ref{co:ub-dropping} is still a valid upper bound, the age for M/G/1 systems becomes a lower bound to age for G/G/1 systems. In fact, this can be easily proved by slightly modifying the proof of Theorem~\ref{th:mg-dropping}.

\begin{figure}
	\centering
	\includegraphics[width=3.3in]{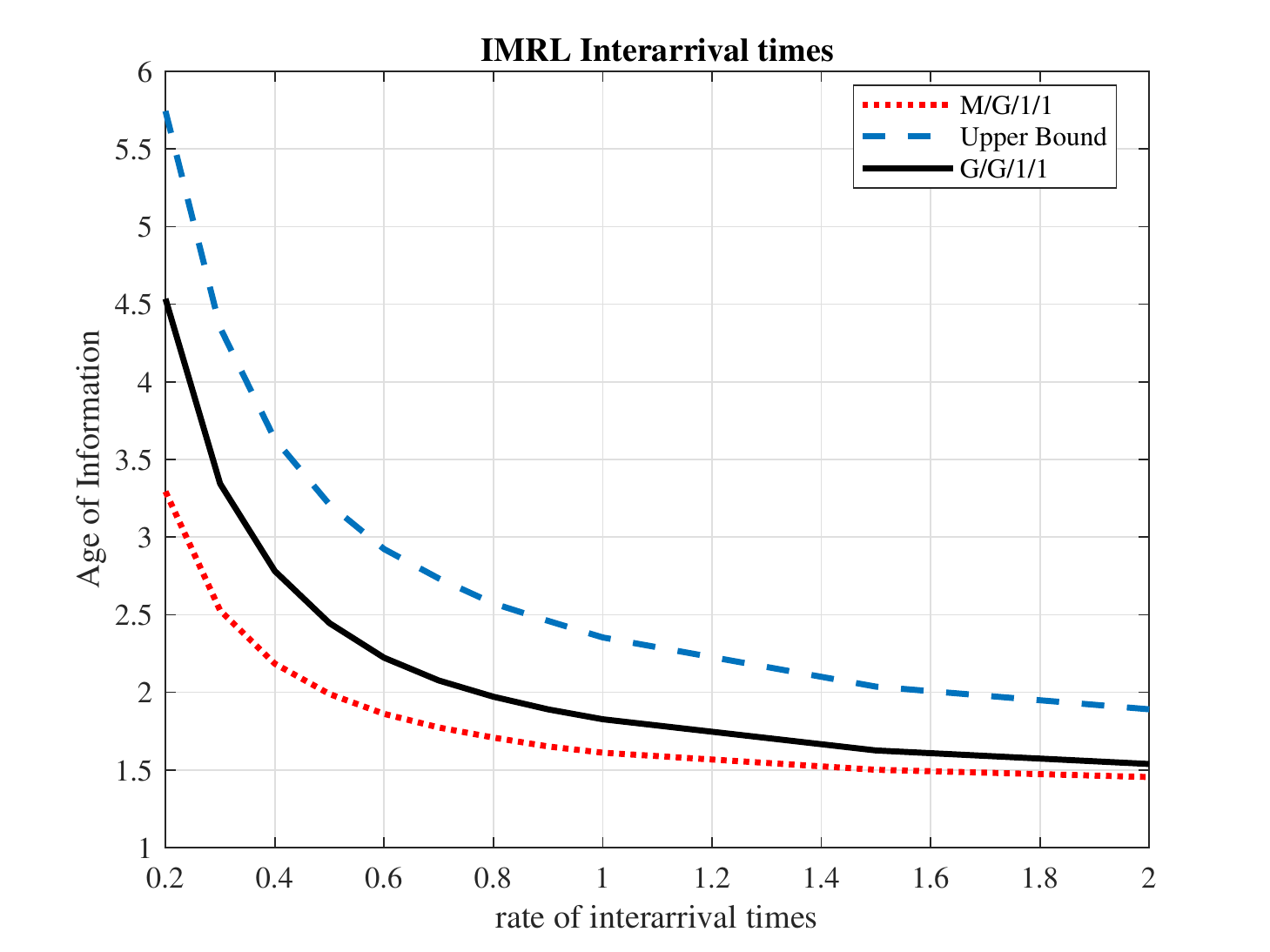}
	\caption{IMRL interarrival times for G/G/1/1 with dropping.}\label{fig/imrl}
\end{figure}
\section{G/G/1/1 with preemption in service}

For preemption in service discipline, average age can also be written as the difference of the areas of two triangles, divided by the expected value of the effective interarrival time. From Fig.~\ref{fig/geometry-preemption}, we have
\begin{align}
\Delta_\text{GG} =& \frac{E[({R}_n+\tilde{S}_{n+1})^2] - E[(\tilde{S}_{n+1})^2]}{2E[{R}_n]}   \\
=& \frac{E[{R}^2]}{2E[{R}]} + E[\tilde{S}]
\label{eq:age-preemption}
\end{align}
where $\tilde{S}_{n+1} = \{S|S<Y\}_{n+1}$ is independent of $R_n$, and time indices are dropped. We note that $R_{n}$ and $G_n$ are equally distributed, then, (\ref{eq:age-preemption}) becomes
\begin{align}
\Delta_\text{GG} = \frac{E[{G}^2]}{2E[{G}]} + E[\tilde{S}].
\label{eq:age-preemption2}
\end{align}
It is important to note that the random variable $G$ in this model is not the same $G$ as in dropping model. The difference can be observed from Figs.~\ref{fig/geometry-dropping} and \ref{fig/geometry-preemption} by noting the change in scale for $S_n$. In the following, we derive an exact closed form expression for (\ref{eq:age-preemption2}), which is valid for arbitrary interarrival and service times. Unlike the expression in dropping model, the expression here does not require the calculation of the average number of arrivals between two consecutive effective interarrivals.

\begin{theorem}
Consider a G/G/1/1 system with preemption in service, where $Y_n$ are interarrival times and $S_n$ are service times. The average age of this system is given by
\begin{align}
\Delta_\text{GG} =& \frac{E[Y^2]}{2E[Y]} + \frac{E[Y \bar F_S(Y)]}{E[\bar F_S(Y)]} + E[\tilde{S}]  
\label{eq:age-exact-preemption}
\end{align}
where $\bar F_S(\cdot)$ is the complementary cdf of $S$, and $\tilde{S}\!= \!S|S<Y$.
\label{th:exact-preemption}
\end{theorem}
\begin{proof}
We know from our system model that $G=\sum_{k=1}^K Y_k$ is a random sum of random numbers, where $K$ is a geometric random variable. From Wald's equation \cite[Theorem 3.3.2]{book/Ross}, we have $E[G] = E[K]E[Y]$. Next, we derive an expression for the second moment of effective interarrival times, $E[G^2]$. Let us first use the indicator function in (\ref{eq:indicator}) and the expansion of $E[G^2]$ in (\ref{eq:7}). Similar to dropping in service discipline, $I_k$ is independent of $Y_k$. Let us consider
\begin{align}
\sum_{l=1}^{k-1} E [I_k Y_l I_l] \! &= \sum_{l=1}^{k-1} \!E [Y_l | I_k = 1]\text{Pr}(I_k=1) \\
&= (k-1) E [Y | Y < S] E[I_k]
\end{align} 
where we used the fact that conditions $I_k = 1$ and $Y_l < S$ are equivalent for preemption in service discipline and for $l <k$. It can be shown that $\sum_{k=1}^\infty E[I_k] = E[K]$ and $\sum_{k=1}^\infty E[I_k] \sum_{l=1}^{k-1}   E[I_l] \leq \frac{1}{2}E[K(K-1)]$. We have
\begin{align}
\hspace{-6pt}E[ G^2] = & E[Y^2] E[K]  + E[Y] E[Y|Y<S] E[K(K-1)]. \label{eq:dummy-preemptive}
\end{align}
Next, using Bayes' rule, we can calculate that 
\begin{align}
E [Y | Y < S] &= \int_0^\infty y \frac{\text{Pr}(Y < S | Y = y)}{\text{Pr}(Y<S)}f_Y(y)dy \\
&= \frac{E[Y\bar F_S(Y)]}{1-p}.
\end{align}
where, $\text{Pr}(Y < S | Y = y) = \text{Pr}(S >y) = \bar F_S(y)$. Now, the average age can be written as
\begin{align}
\Delta_\text{GG} = \frac{E[Y^2]}{2E[Y]} + E[Y\bar F_S(Y)]\frac{E[K(K-1)]}{2E[K](1-p)} + E[\tilde{S}].
\end{align}
Since $K$ is geometric with $p= E[\bar F_S(Y)]$, we have (\ref{eq:age-exact-preemption}).
\end{proof}

An easier to calculate upper bound is given in the following corollary. 
\begin{corollary}
Consider a G/G/1/1 system with preemption in service, where $Y_n$ are interarrival times and $S_n$ are service times. The average age of this system is always upper bounded by
\begin{align}
\Delta_\text{GG} \leq& \frac{E[Y^2]}{2E[Y]} + \frac{E[Y] (1-E[\bar F_S(Y))]}{E[\bar F_S(Y)]} + E[\tilde{S}]  
\label{eq:age-ub-preemption}
\end{align}
where $\bar F_S(\cdot)$ is the complementary cdf of $S$, and $\tilde{S} \!=\! S|S<Y$.
\label{co:ub-preemption}
\end{corollary}

The proof of Corollary~\ref{co:ub-preemption} follows by noting that $E[Y|Y<S] \leq E[Y]$. The upper bound in (\ref{eq:age-ub-preemption}) is achieved when $K$ is independent of $Y_k$. An example of this is the multicast model in \cite{Yates-multicast}, where random sum parameter $K$ is independent of $Y_k$. 

In order to observe the tightness of the bound in Corollary~\ref{co:ub-preemption} for a general case, we simulate the same G/G/1/1 system as in the case with dropping discipline, calculate its age using Theorem~\ref{th:exact-preemption} and compare it to the upper bound in Corollary~\ref{co:ub-preemption}. In Fig.~\ref{fig/SE_pre}, we observe that the difference between the exact age and the upper bound is bounded and small. We also observe that the difference between the exact age and the upper bound depends on the interarrival and service time distributions. The upper bound is tighter for uniform interarrival time and Rayleigh service time distributions than it is for shifted exponential interarrival and service times. 

In addition, age curve with respect to the rate parameter of the interarrival times is not monotonic. When $\lambda$ is very large, in other words when the interarrivals are too frequent, preemption starts to overload the system. Time duration between two successive successful interarrivals gets larger, and hence age increases. This observation for G/G/1/1 systems with preemption in service differs significantly from M/M/1/1 systems with preemption in service, where age is monotonically decreasing in $\lambda$ \cite{YatesKaul-ArXiv}. In addition, minimum age for G/G/1/1 systems over the rate parameter is smaller in dropping scenario than it is in preemption in service scenario. However, we know from \cite{YatesKaul-ArXiv} and \cite{CostaEphr-TransIT} that the opposite is true for M/M/1/1 systems. These observations reassure our initial motivation to consider the AoI for G/G/1 systems, as they can behave much more differently than M/M/1 systems. 

\begin{figure}
	\begin{subfigure}[t]{0.5\textwidth}
		\centering
		\includegraphics[width=3.3in]{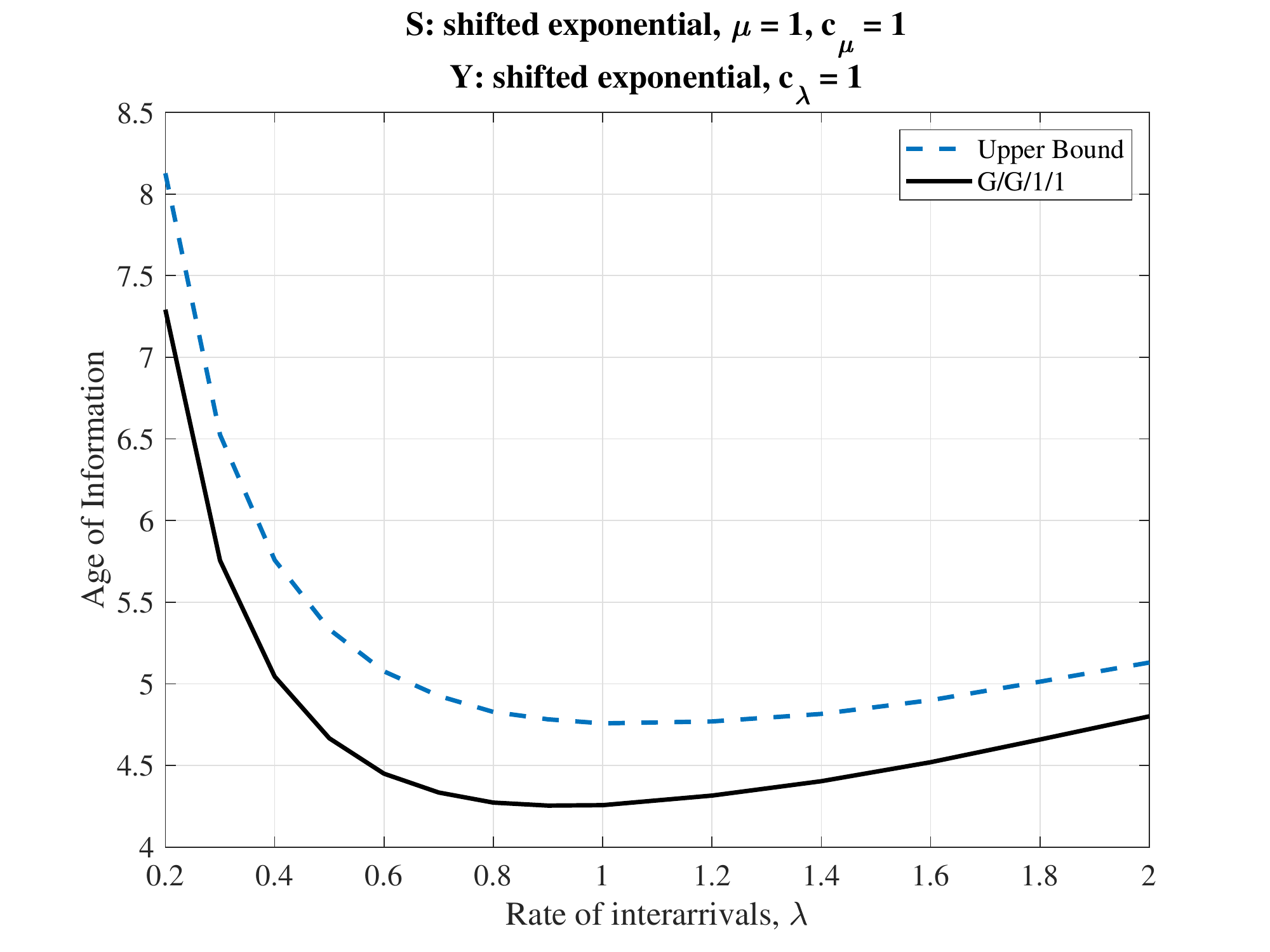}
\vspace{-5pt}
		\caption{}
	\end{subfigure}
    \begin{subfigure}[t]{0.5\textwidth}
		\centering
		\includegraphics[width=3.3in]{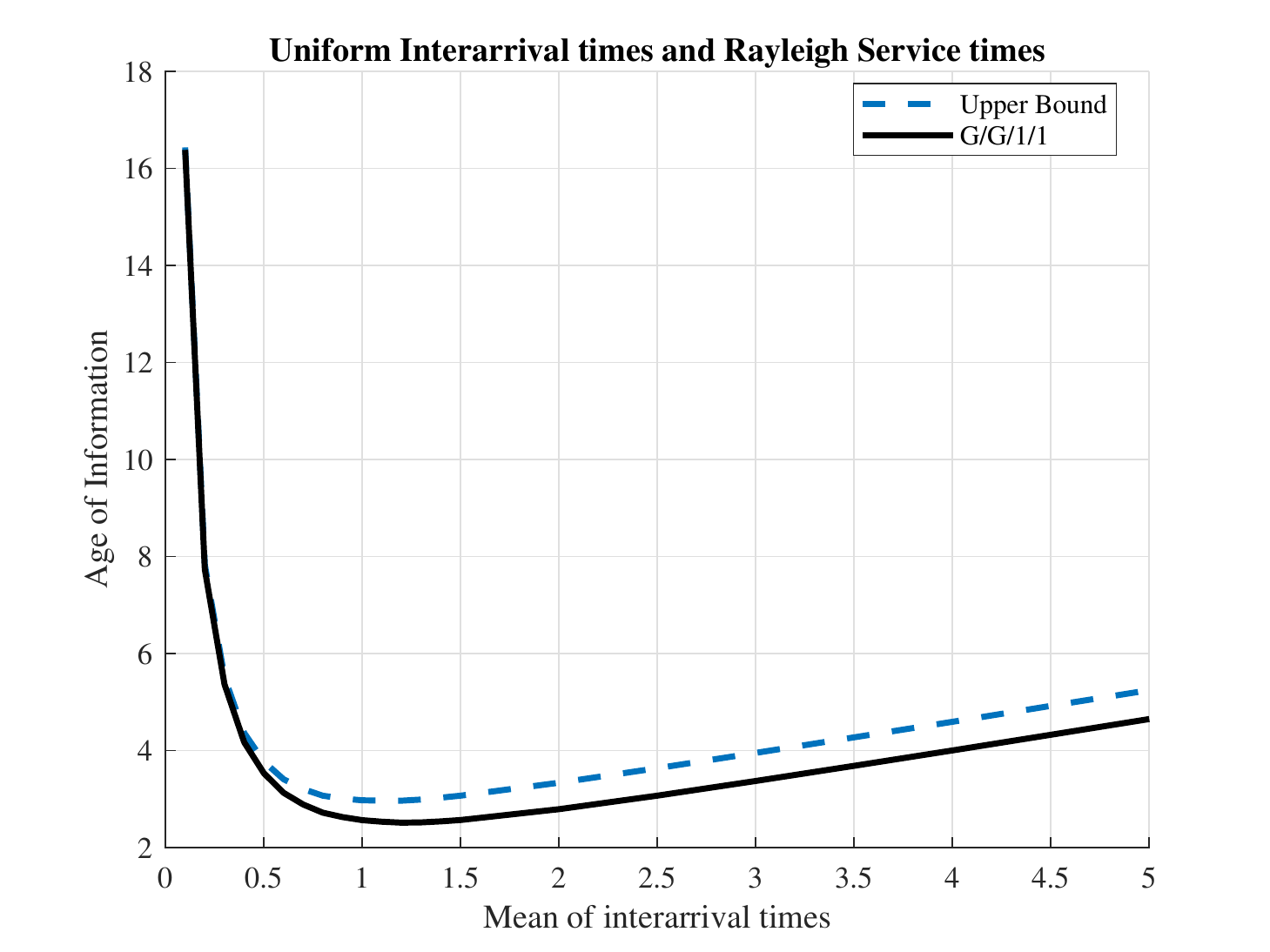}
\vspace{-5pt}
		\caption{}		
	\end{subfigure}
	\caption{Average AoI for G/G/1/1 with preemption in service.}\label{fig/SE_pre}
\end{figure}

\section{Conclusions}
Most real world applications require non-exponential interarrival and service time distributions. This paper is an attempt to extend AoI approach to more practical communication scenarios. We derived exact expressions for and upper bounds to AoI for two service disciplines. We observed that the upper bounds are in general close to exact average age. Designing general communication systems with respect to these upper bounds will result in achievable age values.

\end{document}